\documentclass[11pt]{article}

\usepackage{amsmath,amssymb,amsthm,mathrsfs}
\usepackage{geometry}
\usepackage{mathtools}
\usepackage{bm}
\usepackage{bbm}
 
\newtheorem{definition}{Definition}[section]
\newtheorem{remark}[definition]{Remark}

\newtheorem{proposition}[definition]{Proposition}
\newtheorem{problem}[definition]{Problem}

\newcommand{\A}{\mathcal A}
\newcommand{\Hh}{\mathcal H}

\newcommand{\D}{\mathcal D}
\newcommand{\Dt}{\widetilde{\mathcal D}}
\newcommand{\id}{\mathbbm 1}
\newcommand{\ent}{\mathbin{\rotatebox[origin=c]{90}{\scriptsize$\boldsymbol{\circlearrowleft}$}}}
 \newcommand{\ud}{\mathrm d}
 \newcommand{\ii}{\mathrm i}
 \newcommand{\ex}{\mathrm e}
\usepackage{xcolor}
\usepackage[
  colorlinks=true,
  linkcolor=blue!55!black,
  citecolor=blue!55!black,
  urlcolor=blue!55!black,
  breaklinks=true,
  pdfauthor={Jean-Pierre Gazeau},
  pdftitle={Bi-Oriented Interlaced Matrix Multiplication and an Operator Factorisation of the Quantum Harmonic Oscillator}
]{hyperref}

\title{Bi-Oriented Interlaced Matrix Multiplication and an\\
Operator Factorisation of the Quantum Harmonic Oscillator}

\author{Jean-Pierre Gazeau\thanks{Universit\'e Paris Cit\'e, CNRS,
Astroparticule et Cosmologie, F-75013 Paris, France, and Faculty of
Mathematics, University of Bia\l ystok, 15-245 Bia\l ystok, Poland.
Email: \texttt{gazeau@apc.in2p3.fr; j.gazeau@uwb.edu.pl}}}

\date{\today}

\begin{document}

\maketitle

\begin{abstract}
We introduce a bi-oriented interlaced matrix product on $2\times2$
matrices over a unital, possibly non-associative algebra, extending a
matrix mnemonic for Cayley--Dickson doubling. The opposite orientations
$\diamond_R$ and $\diamond_L$ coincide precisely when the underlying
algebra is commutative. For the canonical position and momentum operators
$Q,P$ on $L^2(\mathbb R)$, a matrix $\mathcal D$ and its companion
$\widetilde{\mathcal D}$ satisfy, on the common Schwartz domain,
$$
\widetilde{\mathcal D}\diamond_R\mathcal D
=\widetilde{\mathcal D}\diamond_L\mathcal D
=2H\,\id_2,\qquad H=\tfrac12(P^2+Q^2).
$$
We study the associated interlaced eigenvalue equations for
operator-valued columns. The combinations $\Phi=A+\ii B$ and
$\Psi=A-\ii B$ decouple the equations into two-sided Sylvester equations.
We construct explicit rank-one generalized eigenoperators from position
and momentum eigenstates. Their distributional kernels belong to the
doubled rigged Hilbert space
$\mathcal S(\mathbb R^2)\subset L^2(\mathbb R^2)
\subset\mathcal S'(\mathbb R^2)$; the corresponding position and momentum
operators act on different tensor factors and therefore commute.
The rank-one families do not exhaust the distributional solutions:
higher-order generalized solutions are also exhibited. We distinguish
the algebraic identities on the Schwartz core from questions concerning
closed-operator realizations.
\end{abstract}




\section{Introduction}

This note is offered as a contribution to the topical collection in
memory of Franciszek Hugon Szafraniec. The author had the privilege of
writing two papers with him, on coherent states built from holomorphic
Hermite polynomials and the non-commutative plane~\cite{GazeauSzafraniec2011},
and on the quantum angle operator associated with a number
operator~\cite{GazeauSzafraniec2016}. In both, Szafraniec's insistence on
precise operator-theoretic foundations --- domains, closability,
self-adjointness --- shaped the final result at least as much as the
physical motivation did; see in particular his own rigorous analytic
models for the quantum harmonic oscillator~\cite{Szafraniec1998}. The
present note is written in that spirit: its central object is again the
harmonic oscillator Hamiltonian $H=\tfrac12(P^2+Q^2)$, approached here
through a new algebraic factorisation device.

The octonions and, more generally, the Cayley--Dickson tower may be
written in terms of $2\times2$ block matrices with entries in the
previous algebra of the tower. In that setting, multiplication can be
encoded by a rule in which certain products are taken in their usual
order while others are reversed~\cite{Gazeau2026}; see
\cite{Baez,Schafer,Okubo,SpringerVeldkamp} for background on octonions
and non-associative algebras.

The point of the present note is that this rule should be viewed as an
algebraic object in its own right, independently of octonions. It admits
two natural orientations, $\diamond_R$ and $\diamond_L$ (the first one was denoted by $\ent$ in \cite{Gazeau2026}), to be studied as a pair (Section~\ref{sec:products}), and it applies to matrices with
entries in \emph{any} unital algebra, in particular to matrices of
operators on a Hilbert space. Taking the entries to be the position and
momentum operators produces the main new results of this note:

\begin{itemize}
\item an interlaced factorisation $\widetilde{\mathcal D}\diamond_{R}\mathcal D
=\widetilde{\mathcal D}\diamond_{L}\mathcal D=2H\,\id_2$ of the harmonic
oscillator Hamiltonian (Section~\ref{sec:factorisation}, Proposition~\ref{prop:squares});                            
\item explicit generalized solutions of the interlaced eigenvalue problem
$\mathcal D\diamond_R V=\lambda V$ for operator-valued column vectors $V$,
including rank-one operators built from the
generalised position and momentum eigenstates of a rigged Hilbert space
(Section~\ref{sec:eigen}, Proposition~\ref{prop:eigPhi}), with the
opposite orientation $\diamond_L$ producing the mirror-image
(ket$\leftrightarrow$bra exchanged) family for the same eigenvalues; we
show this is no accident, but the trace of a two-mode eigenvalue problem
hiding inside the one-sided-looking equation $\mathcal D\diamond_\sigma
V=\lambda V$.
\end{itemize}

Section~\ref{sec:products} sets up the algebraic framework;
Section~\ref{sec:structure} records a few structural facts needed later;
Sections~\ref{sec:factorisation}--\ref{sec:eigen} contain the operator-theoretic
core of the note; Section~\ref{sec:outlook} closes with an outlook. 

Although most of the results are elementary, we have chosen to provide explicit proofs for the sake of completeness and clarity.


\section{The two oriented interlaced products}
\label{sec:products}

Let $\A$ be a unital algebra, not assumed commutative and not necessarily
associative.  For
$$
X=
\begin{pmatrix}
A&B\\
C&D
\end{pmatrix},
\qquad
Y=
\begin{pmatrix}
E&F\\
G&H
\end{pmatrix},
$$
with entries in $\A$, define the \emph{right-oriented interlaced product}
$$
\boxed{
X\diamond_R Y
=
\begin{pmatrix}
AE+GB & FA+BH\\
EC+DG & CF+HD
\end{pmatrix}.
}
$$

Reading each entry as a pair (first summand, second summand) the pattern
of orderings is
$$
(1,1): (X_{\cdot 1})(Y_{1\cdot}) + (Y_{2\cdot})(X_{\cdot 2})
\quad\Rightarrow\quad
\text{natural then reversed},
$$
and similarly for the other entries.  Labelling natural order as
$\mathrm{R}$ (right factor of $X$ contracts first) and reversed order as
$\mathrm{L}$, the block pattern is
$$
\boxed{
\begin{array}{cc}
RL&LR\\
LR&RL
\end{array}
}
\qquad
\text{or equivalently}\qquad
\boxed{\mathrm{RLLR}/\mathrm{LRRL}}.
$$


The opposite, or \emph{left-oriented interlaced product} is
$$
\boxed{
X\diamond_L Y
=
\begin{pmatrix}
EA+BG & AF+HB\\
CE+GD & FC+DH
\end{pmatrix}.
}
$$

Its orientation pattern is
$$
\boxed{
\begin{array}{cc}
LR&RL\\
RL&LR
\end{array}
}
\qquad
\text{or equivalently}\qquad
\boxed{\mathrm{LRRL}/\mathrm{RLLR}}.
$$

\begin{remark}[Relation between $\diamond_R$ and $\diamond_L$]
\label{rem:swap}
A direct inspection of the definitions shows that
$$
X\diamond_L Y = X\diamond_R Y
$$
holds \emph{if and only if\/} $\A$ is commutative.  In general, the two
products are distinct.  The relation that does hold unconditionally is
purely formal, and needs no enlargement of the algebraic setting to
state: \emph{if, in every entry of the formula defining $X\diamond_R Y$,
one reverses the order of each individual product ($pq\mapsto qp$, using
the multiplication of $\A$), one obtains exactly the formula defining
$X\diamond_L Y$} --- with $X$ and $Y$ playing the same roles throughout,
neither matrix being swapped or otherwise transformed. This is a direct,
entry-by-entry symmetry between the two boxed definitions above, and can
be checked term by term.
\end{remark}

\begin{definition}
The triple
$$
\mathfrak I(\A)
=
\bigl(M_2(\A),\diamond_R,\diamond_L\bigr)
$$
will be called the \emph{bi-oriented interlaced matrix algebra} over $\A$.
The term ``algebra'' is used loosely: in general neither $\diamond_R$ nor
$\diamond_L$ is associative.
\end{definition}

%

The identity matrix $\id_2$ is a two-sided unit
for both interlaced products:
$$
\id_2\diamond_{R,L}X
=
X\diamond_{R,L}\id_2
=
X,
\qquad X\in M_2(\A).
$$
This follows immediately from the defining formulas.

\section{Further structural properties}
\label{sec:structure}

The two products are not independent; they represent opposite
orientations of the same interlacing principle, formally exchanged by an
involution $\mathcal O$ with $\mathcal O^2=\id$, $\mathcal
O(\diamond_R)=\diamond_L$, $\mathcal O(\diamond_L)=\diamond_R$. By
Remark~\ref{rem:swap}, $\mathcal O$ acts entrywise, by reversing the
order of every individual product appearing in the defining formula. The
natural object is therefore the pair $(\diamond_R,\diamond_L)$, not a
single product.

For operator entries, the two orientations interact with the adjoint
exactly as ordinary matrix multiplication does, \emph{provided} the
adjoint of a $2\times2$ block operator matrix is taken correctly: not
merely by conjugating each entry in place, but by also swapping the
off-diagonal block positions, as it must for any operator on
$\Hh\oplus\Hh$.

\begin{proposition}[Adjoint interchange]
\label{prop:adjoint}
Let $\Hh$ be a Hilbert space and $X=\begin{pmatrix}A&B\\C&D\end{pmatrix}$,
$Y=\begin{pmatrix}E&F\\G&H\end{pmatrix}\in M_2(B(\Hh))$, regarded as
operators on $\Hh\oplus\Hh$~\cite{KadisonRingrose1}. Write
$$
X^*=\begin{pmatrix}A^*&C^*\\B^*&D^*\end{pmatrix}
$$
for the Hilbert-space adjoint of $X$ --- entries conjugated \emph{and}
off-diagonal positions swapped, exactly as for any block operator matrix
--- and similarly for $Y^*$. Then, for either orientation $\sigma\in\{R,L\}$,
$$
(X\diamond_\sigma Y)^* = Y^*\diamond_\sigma X^*.
$$
In particular, if $X$ and $Y$ are self-adjoint as block operators (which
requires $A^*=A$, $D^*=D$, and $C=B^*$ --- not merely $B=C$), then
$$
(X\diamond_\sigma Y)^* = Y\diamond_\sigma X.
$$
\end{proposition}

\begin{proof}
We give the case $\sigma=R$; the case $\sigma=L$ is proved the same way.
Conjugating each entry of $X\diamond_R Y$ and swapping the off-diagonal
block positions gives
$$
(X\diamond_R Y)^*
=
\begin{pmatrix}
E^*A^*+B^*G^* & C^*E^*+G^*D^*\\
A^*F^*+H^*B^* & F^*C^*+D^*H^*
\end{pmatrix}.
$$
On the other hand, applying the defining formula of $\diamond_R$ to
$Y^*=\begin{pmatrix}E^*&G^*\\F^*&H^*\end{pmatrix}$ and
$X^*=\begin{pmatrix}A^*&C^*\\B^*&D^*\end{pmatrix}$ --- note the swapped
off-diagonal entries in each --- gives, entry by entry,
$$
Y^*\diamond_R X^*
=
\begin{pmatrix}
E^*A^*+B^*G^* & C^*E^*+G^*D^*\\
A^*F^*+H^*B^* & F^*C^*+D^*H^*
\end{pmatrix},
$$
which coincides with $(X\diamond_R Y)^*$ term by term. The particular
case follows since $X^*=X$, $Y^*=Y$ under the stated self-adjointness
hypotheses.
\end{proof}

Consequently, for self-adjoint $X,Y$ the symmetrised product
$X\circ_\sigma Y=\tfrac12(X\diamond_\sigma Y+Y\diamond_\sigma X)$ is
itself self-adjoint:
$$
(X\circ_\sigma Y)^*
=\tfrac12\bigl[(X\diamond_\sigma Y)^*+(Y\diamond_\sigma X)^*\bigr]
=\tfrac12\bigl[Y\diamond_\sigma X+X\diamond_\sigma Y\bigr]
=X\circ_\sigma Y,
$$
a genuine Jordan-algebra-type property, in the spirit of the
symmetrised operator product of~\cite{JordanVNW,McCrimmon}; a fuller
comparison with Jordan algebras is left for future work.

Beyond the adjoint, the bi-oriented structure supports left and right
translations $L_X^{R}(Y)=X\diamond_R Y$, $R_X^{R}(Y)=Y\diamond_R X$ (and
similarly for $\diamond_L$), and, correspondingly, four notions of
one-sided inverse and loop-like structure~\cite{Pflugfelder}. Likewise,
since two products coexist, associativity defects come in four flavours
$\mathcal A_{\sigma\tau}(X,Y,Z)=(X\diamond_\sigma Y)\diamond_\tau Z-
X\diamond_\sigma(Y\diamond_\tau Z)$, $\sigma,\tau\in\{R,L\}$; already for
$\sigma=\tau$ and $\A$ associative, $\mathcal A_{RR}$ and $\mathcal A_{LL}$
do not vanish identically, so the interlaced products remain
non-associative even over an associative base algebra. We record a single
representative problem and otherwise leave this line of investigation for
a separate, purely algebraic study.

\begin{problem}
Classify the subspaces of $M_2(\A)$ on which one or several of the four
associators $\mathcal A_{RR},\mathcal A_{RL},\mathcal A_{LR},\mathcal
A_{LL}$ vanish, and identify the interlaced analogue of the alternative
identities $[x,x,y]=0$, $[y,x,x]=0$.
\end{problem}


\section{An interlaced factorisation of $P^2+Q^2$}
\label{sec:factorisation}

Let $Q$ and $P$ be the usual position and momentum operators on
$L^2(\mathbb R)$,
$$
(Qf)(x)=xf(x),
\qquad
(Pf)(x)=-\ii\frac{\ud}{\ud x}f(x),
$$
initially defined, and essentially self-adjoint, on the Schwartz space
$\mathcal S(\mathbb R)$~\cite{ReedSimon2}, on which they act invariantly
and satisfy the canonical commutation relation $[Q,P]=\ii\id_{\mathcal{H}}$.  Define
$$
\D=
\begin{pmatrix}
Q&-P\\
P&Q
\end{pmatrix},
\qquad
\Dt=
\begin{pmatrix}
Q&P\\
-P&Q
\end{pmatrix}.
$$

\begin{proposition}[Interlaced squares]
\label{prop:squares}
On $\mathcal S(\mathbb R)\oplus\mathcal S(\mathbb R)$,
$$
\D\diamond_R\D
=
\begin{pmatrix}
Q^2-P^2&-2PQ\\
2QP&Q^2-P^2
\end{pmatrix},
\qquad
\D\diamond_L\D
=
\begin{pmatrix}
Q^2-P^2&-2QP\\
2PQ&Q^2-P^2
\end{pmatrix}.
$$
The mixed products, by contrast, are orientation-independent:
$$
\Dt\diamond_R\D
=
\Dt\diamond_L\D
=
(Q^2+P^2)\,\id_2.
$$
\end{proposition}

\begin{proof}
The computations of $\D\diamond_{R,L}\D$ are straightforward applications
of the definitions.  We detail only the mixed product $\Dt\diamond_R\D$,
with $\Dt=\bigl(\begin{smallmatrix}Q&P\\-P&Q\end{smallmatrix}\bigr)$
playing the role of $X=\bigl(\begin{smallmatrix}A&B\\C&D\end{smallmatrix}\bigr)$
and $\D$ the role of $Y=\bigl(\begin{smallmatrix}E&F\\G&H\end{smallmatrix}\bigr)$,
so that $A=D=Q$, $B=P$, $C=-P$, $E=H=Q$, $F=-P$, $G=P$.

The $(1,1)$ entry is $AE+GB = Q^2+P^2$; the $(1,2)$ entry is
$FA+BH=-PQ+PQ=0$; the $(2,1)$ entry is $EC+DG=-QP+QP=0$; the $(2,2)$
entry is $CF+HD=P^2+Q^2$. The computation for $\Dt\diamond_L\D$ gives the
same result, the roles of $QP$ and $PQ$ in the off-diagonal entries being
exchanged but their sum still vanishing.
\end{proof}

Setting $H=\tfrac12(P^2+Q^2)$ for the harmonic oscillator Hamiltonian,
$$
\boxed{
\Dt\diamond_R\D
=
\Dt\diamond_L\D
=
2H\,\id_2.
}
$$
This is the interlaced analogue of the classical factorisation
$(x-\ii y)(x+\ii y)=x^2+y^2$, with the ordinary scalar product replaced by the
bi-oriented interlaced matrix product; here it is orientation-independent
even though $\D\diamond_R\D\neq\D\diamond_L\D$. Accordingly, $\D$ may be
regarded as an \emph{interlaced square root} of $2H\,\id_2$, and, in this
precise sense, as an interlaced Dirac-type operator for the harmonic
oscillator.

\begin{remark}[On rigour]
\label{rem:rigour}
As $Q$ and $P$ are unbounded, Proposition~\ref{prop:squares} is, strictly
speaking, an identity of operators restricted to
$\mathcal S(\mathbb R)\oplus\mathcal S(\mathbb R)$, the natural common
invariant core on which $\D$ and $\Dt$ act and on which $Q,P$ are
essentially self-adjoint~\cite{ReedSimon2}. Promoting it to an identity
between closed operators on $L^2(\mathbb R)\oplus L^2(\mathbb R)$, and
settling the self-adjointness or normality of $\D$ itself with respect to
a suitable Hilbert-space structure, is exactly the kind of question that
a fully rigorous treatment must address; it is in the spirit of
Szafraniec's own operator models for the harmonic
oscillator~\cite{Szafraniec1998} and of our joint constructions of
rigorously defined coherent-state and angle operators associated with
$H$~\cite{GazeauSzafraniec2011,GazeauSzafraniec2016}. We record the two
natural questions as open problems.
\end{remark}

\begin{problem}
\label{prob:selfadjoint}
Is $\D$ essentially self-adjoint, skew-adjoint, or normal with respect to
a natural Hilbert-space structure on $L^2(\mathbb R)\oplus L^2(\mathbb R)$?
\end{problem}

\begin{problem}
\label{prob:closure}
Can the interlaced identity $\Dt\diamond_R\D=2H\,\id_2$ be promoted from a
formal identity on $\mathcal S(\mathbb R)\oplus\mathcal S(\mathbb R)$ to
an identity between closed operators on $L^2(\mathbb R)\oplus L^2(\mathbb R)$?
\end{problem}


\section{Creation and annihilation operators}
\label{sec:ladder}

Introduce the standard creation and annihilation operators
$$
a=\frac{Q+\ii P}{\sqrt2},
\qquad
a^\dagger=\frac{Q-\ii P}{\sqrt2},
$$
satisfying $[a,a^\dagger]=\id_{\mathcal{H}}$ and $H=a^\dagger a+\tfrac12 \id_{\mathcal{H}}$.
In terms of $a$ and $a^\dagger$,
$$
\D
=
\frac{1}{\sqrt2}
\begin{pmatrix}
a+a^\dagger & \ii(a-a^\dagger)\\
-\ii(a-a^\dagger) & a+a^\dagger
\end{pmatrix},
\qquad
\Dt
=
\frac{1}{\sqrt2}
\begin{pmatrix}
a+a^\dagger & -\ii(a-a^\dagger)\\
\ii(a-a^\dagger) & a+a^\dagger
\end{pmatrix}.
$$
The interlaced identity $\Dt\diamond_R\D = 2H\id_2$ then becomes, after
substitution,
$$
\Dt\diamond_R\D
=
(2a^\dagger a + \id_{\mathcal{H}})\,\id_2,
$$
which exhibits $\D$ as the ``square root'' of the number operator matrix
$N \id_2$ (where $N=a^\dagger a$) in the bi-oriented interlaced sense, in
analogy with the Dirac square root of the Klein--Gordon
operator~\cite{Dirac}. A systematic representation of the
Weyl--Heisenberg algebra through $\diamond_{R,L}$ remains to be developed.


\section{The operator eigenvalue problem for $\mathcal{D}$}
\label{sec:eigen}

We now study the interlaced eigenvalue equation for the matrix $\D$,
where the ``eigenvectors'' are themselves $2\times1$ column matrices
of operators. Both orientations turn out to reduce to Sylvester-type
equations, solved by rank-one operators built from the (generalised)
position and momentum eigenstates of a rigged Hilbert space; this is
the second main result of the contribution.

\subsection*{Extension of $\diamond_{R,L}$ to column vectors}

For a $2\times1$ column $V=\bigl(\begin{smallmatrix}A\\B\end{smallmatrix}\bigr)$
with operator entries $A,B$, we extend $\diamond_R$ and $\diamond_L$ by
setting the second column of the right factor to zero in the $2\times2$
formulae.  A direct reading of the definitions gives
\begin{align*}
\D\diamond_R V &=
\begin{pmatrix}
QA - BP\\
AP + QB
\end{pmatrix},
&
\D\diamond_L V &=
\begin{pmatrix}
AQ - PB\\
PA + BQ
\end{pmatrix}.
\end{align*}
The difference between the two orientations is now clean and uniform:
under $\diamond_R$, $Q$ multiplies \emph{from the left} and $P$
multiplies \emph{from the right}, whichever of $A,B$ it acts on
($QA$, $-BP$, $AP$, $QB$); under $\diamond_L$ the two roles are
exchanged, $P$ always from the left and $Q$ always from the right
($AQ$, $-PB$, $PA$, $BQ$). This sharpens, for the vector case, the same
left/right asymmetry that underlies the RLLR and LRRL
patterns of Section~\ref{sec:products}.

\subsection*{The two eigenvalue systems}

The interlaced eigenvalue equations
$\D\diamond_R V = \lambda V$ and $\D\diamond_L V = \lambda V$,
for $\lambda\in\mathbb{C}$, yield the operator systems:
$$
(\diamond_R):\quad
\begin{cases}
QA - BP = \lambda A,\\
AP + QB = \lambda B,
\end{cases}
\qquad
(\diamond_L):\quad
\begin{cases}
AQ - PB = \lambda A,\\
PA + BQ = \lambda B.
\end{cases}
$$

\subsection*{Complexification}

Set $\Phi = A + \ii B$, $\Psi = A - \ii B$, so $A=(\Phi+\Psi)/2$,
$B=(\Phi-\Psi)/(2\ii)$. Combining row~1 and $\ii\cdot$row~2 of each
system (for $\Phi$), and row~1 minus $\ii\cdot$row~2 (for $\Psi$), gives
after a direct computation
$$
(\diamond_R):\quad Q\Phi+\ii\,\Phi P=\lambda\Phi,\qquad Q\Psi-\ii\,\Psi P=\lambda\Psi,
$$
$$
(\diamond_L):\quad \Phi Q+\ii\,P\Phi=\lambda\Phi,\qquad \Psi Q-\ii\,P\Psi=\lambda\Psi.
$$
Equivalently, in terms of the ladder operators $a=(Q+\ii P)/\sqrt2$,
$a^\dagger=(Q-\ii P)/\sqrt2$,
\begin{equation}
\label{eq:eigR}
\boxed{
(\diamond_R):\quad
\tfrac{1}{\sqrt2}\bigl(\{\Phi,a\}-[\Phi,a^\dagger]\bigr)=\lambda\Phi,
\qquad
\tfrac{1}{\sqrt2}\bigl(\{\Psi,a^\dagger\}-[\Psi,a]\bigr)=\lambda\Psi,
}
\end{equation}
\begin{equation}
\label{eq:eigL}
\boxed{
(\diamond_L):\quad
\tfrac{1}{\sqrt2}\bigl(\{\Phi,a\}+[\Phi,a^\dagger]\bigr)=\lambda\Phi,
\qquad
\tfrac{1}{\sqrt2}\bigl(\{\Psi,a^\dagger\}+[\Psi,a]\bigr)=\lambda\Psi.
}
\end{equation}
In both orientations $\Phi$ and $\Psi$ are \emph{uncoupled from each
other} --- each equation involves only one of the two --- but neither
reduces to a one-sided eigenvalue equation of the familiar $a\Phi=\mu\Phi$
type: each is a genuine \emph{Sylvester equation}, with one of $Q,P$
acting on $\Phi$ (or $\Psi$) from the left and the other from the right.
The two orientations are mirror images of one another under left/right
exchange: $\diamond_R$ has $Q$ on the left and $P$ on the right;
$\diamond_L$ has the roles reversed.

\subsection*{Solution via the rigged Hilbert space}

For the generalized position and momentum eigenstates we use the
one-particle Gelfand triplet
$$
\mathcal S(\mathbb R)\subset\Hh=L^2(\mathbb R)
\subset\mathcal S'(\mathbb R),
$$
with its generalised position and momentum eigenkets
$|q\rangle,|p\rangle\in\mathcal S'(\mathbb R)$, $q,p\in\mathbb R$,
satisfying $Q|q\rangle=q|q\rangle$ and $\langle p|P=p\langle p|$ in the
distributional
sense~\cite{GelfandVilenkin4,Bohm1978,AntoineTrapani2009,AntoineInoueTrapani2002}.
For the functional-analytic framework of rigged Hilbert spaces and the
algebraic treatment of unbounded operators we refer to the monographs of
Antoine, Trapani and
collaborators~\cite{AntoineTrapani2009,AntoineInoueTrapani2002}.

The Schwartz kernel theorem provides the natural framework for the
operator-valued eigenvalue problem considered here~\cite{Treves1967}:
every continuous linear map $T:\mathcal S(\mathbb R)\to\mathcal
S'(\mathbb R)$ is uniquely represented by a tempered distribution
$K_T\in\mathcal S'(\mathbb R^2)$ through
$$
\langle T\varphi,\psi\rangle
=
\langle K_T,\psi\otimes\varphi\rangle,
\qquad
\varphi,\psi\in\mathcal S(\mathbb R).
$$
Consequently, the operator-valued eigenvalue problem can be formulated
as a distributional differential equation on $\mathbb R^2$: writing
$\Phi(x,y)$ for the kernel of an operator $\Phi:\mathcal S(\mathbb
R)\to\mathcal S'(\mathbb R)$, $Q\Phi$ acts on it as multiplication by
$x$ and $\Phi P$ as $\ii\partial_y$.

This correspondence naturally leads to the doubled rigged Hilbert space
$$
\mathcal S(\mathbb R^2)
\subset L^2(\mathbb R^2)
\subset\mathcal S'(\mathbb R^2),
$$
where $L^2(\mathbb R^2)$ is identified with the Hilbert space of
Hilbert--Schmidt operators on $L^2(\mathbb R)$. The generalized
eigenoperators considered below belong to the distributional extension
of this Hilbert space and need not be Hilbert--Schmidt operators:
precisely, the rank-one solutions below are generalized kernels and not
Hilbert--Schmidt, since the corresponding multiplication operators on
$L^2(\mathbb R^2)$ have no nonzero square-integrable eigenfunctions. We
make no claim here about all possible bounded-operator solutions without
additional domain hypotheses.

\begin{proposition}[Rigged eigenoperators of the $\diamond_R$ system]
\label{prop:eigPhi}
Write $\lambda=q+\ii p\in\mathbb C$, $q,p\in\mathbb R$. The rank-one
operators
$$
\Phi=|q\rangle\langle p|,
\qquad
\Psi=|q\rangle\langle -p|
$$
solve $Q\Phi+\ii\Phi P=\lambda\Phi$ and $Q\Psi-\ii\Psi P=\lambda\Psi$
respectively, for every $q,p\in\mathbb R$.
\end{proposition}

\begin{proof}
Since $Q|q\rangle=q|q\rangle$ and $\langle p|P=p\langle p|$,
$$
Q\Phi+\ii\Phi P=q\Phi+\ii p\Phi=(q+\ii p)\Phi=\lambda\Phi;
$$
similarly $\langle -p|P=-p\langle -p|$ gives
$Q\Psi-\ii\Psi P=q\Psi-\ii(-p)\Psi=(q+\ii p)\Psi=\lambda\Psi$.
\end{proof}

\begin{remark}[Uniqueness, and a Jordan-chain caveat]
In kernel form, $Q\Phi+\ii\Phi P=\lambda\Phi$ reads
$\partial_y\Phi(x,y)=(x-\lambda)\Phi(x,y)$. For fixed $x$ the general
solution of this ODE in $y$ is $C(x)\,\ex^{(x-\lambda)y}$, which is
tempered in $y$ only if $\mathrm{Re}(x-\lambda)=0$, i.e.\ $x=q$; hence
$\Phi(x,\cdot)$ vanishes for $x\neq q$, and $\Phi$ is a distribution in
$x$ supported at the single point $q$. The simplest such distributions
are multiples of $\delta(x-q)$, giving exactly $\Phi=c\,|q\rangle\langle
p|$ as above. A distribution supported at a point need not be a pure
delta, however: using $x\,\delta'(x-q)=q\,\delta'(x-q)-\delta(x-q)$, one
checks directly that
$$
\Phi = c_1\,\delta'(x-q)\,\ex^{-\ii py}
     - c_1\,y\,\delta(x-q)\,\ex^{-\ii py}
     + c_0\,\delta(x-q)\,\ex^{-\ii py}
$$
also solves the equation for any $c_0,c_1\in\mathbb C$, and this
Jordan-chain extends to every order $\delta^{(k)}(x-q)$, $k\ge0$. So
$\Phi=|q\rangle\langle p|$ is the simplest, but not the only, solution
for a given $\lambda$; a complete classification is left for future work
(Problem~\ref{prob:classify}).
\end{remark}

\subsection*{Recovering $A$ and $B$}

Taking the rank-one solutions of Proposition~\ref{prop:eigPhi},
$\Phi=\alpha\,|q\rangle\langle p|$ and $\Psi=\beta\,|q\rangle\langle -p|$
for arbitrary $\alpha,\beta\in\mathbb C$, and inverting
$A=(\Phi+\Psi)/2$, $B=(\Phi-\Psi)/(2\ii)$, gives the $\diamond_R$
eigenvector
\begin{equation}
\label{eq:eigvec}
\boxed{
V_{q,p}(\alpha,\beta)
=\frac12\,|q\rangle
\begin{pmatrix}
\alpha\langle p|+\beta\langle -p|\\[2pt]
-\ii\bigl(\alpha\langle p|-\beta\langle -p|\bigr)
\end{pmatrix},
\qquad
\lambda=q+\ii p,\quad q,p\in\mathbb R,\ \alpha,\beta\in\mathbb C.
}
\end{equation}
The operators  $\Phi$ and $\Psi$
carry no bounded/Hilbert--Schmidt normalisability condition to impose on
one another --- neither is normalisable to begin with --- so both survive
independently: the displayed rank-one family for each $\lambda=q+\ii p$ spans a
two-complex-dimensional subspace when $p\ne0$, spanned by $(\alpha,\beta)=(1,0)$ and $(0,1)$,
but it does not exhaust the distributional solution space, as the previous remark shows.
For $p=0$ the two displayed rank-one operators coincide and this subspace
is one-dimensional.

\subsection*{The $\diamond_L$ case, and comparison}

The same argument, with the roles of ket and bra exchanged (since
$\diamond_L$ places $Q,P$ on the opposite side of $\Phi,\Psi$ relative to
$\diamond_R$), gives
$$
\Phi=\alpha\,|p\rangle\langle q|,
\qquad
\Psi=\beta\,|-p\rangle\langle q|
$$
as the rank-one $\diamond_L$ solutions for the same $\lambda=q+\ii p$.
The two orientations are thus honest mirror images of one another ---
$\diamond_R$ fixes the position label on the ket and carries the
eigenvalue's imaginary part on the momentum label of the bra,
$\diamond_L$ does the opposite --- rather than the
``decouples/does-not-decouple'' asymmetry originally conjectured.

\begin{center}
\renewcommand{\arraystretch}{1.5}
\begin{tabular}{lll}
\hline
& $\diamond_R$ & $\diamond_L$\\
\hline
System & $Q\Phi+\ii\Phi P=\lambda\Phi$, & $\Phi Q+\ii P\Phi=\lambda\Phi$,\\
       & $Q\Psi-\ii\Psi P=\lambda\Psi$  & $\Psi Q-\ii P\Psi=\lambda\Psi$\\
$\Phi,\Psi$ coupled? & No (each equation involves only one) & No\\
Type & Sylvester (two-sided) & Sylvester (two-sided)\\
Rank-one solution & $\Phi=|q\rangle\langle p|,\ \Psi=|q\rangle\langle -p|$
                  & $\Phi=|p\rangle\langle q|,\ \Psi=|-p\rangle\langle q|$\\
\hline
\end{tabular}
\end{center}

\subsection*{Physical interpretation: a two-mode reformulation}

Why should $Q$ acting on one side of $\Phi$ and $P$ on the other, with no
cross-term, produce \emph{sharp} simultaneous eigenstates, when the
naïve (but incorrect) one-sided equation $a\Phi=\mu\Phi$ of a first
attempt would instead force $\Phi$ to be a genuine coherent state, fuzzy
in both $Q$ and $P$ by the Heisenberg amount? The two-sided equation is
secretly a two-mode problem in disguise.

Identify a Hilbert--Schmidt operator with its kernel in
$L^2(\mathbb R^2)\simeq\Hh\otimes\Hh$; extend the correspondence to
tempered distribution kernels when the operations are defined.
The generalized vector $|\Phi\rangle\rangle$ is therefore in
$\mathcal S'(\mathbb R^2)$, not generally in $\Hh\otimes\Hh$. Left
multiplication by $Q$ acts as $Q\otimes \id_{\mathcal{H}}$; right multiplication by $P$
acts, after this identification, as $-(\id_{\mathcal{H}}\otimes P)$, the sign coming
from $P^{\!\top}=-P$ (the position-representation transpose of
$P=-\ii\,\ud/\ud x$, reflecting its time-reversal-odd character; both this
and the underlying vectorisation identities are elementary to check
directly). The $\diamond_R$ equation for $\Phi$ becomes exactly
$$
\bigl(Q^{(1)}-\ii P^{(2)}\bigr)\,|\Phi\rangle\rangle=\lambda\,|\Phi\rangle\rangle,
$$
a genuine two-mode eigenvalue problem, $Q^{(1)}$ acting on the first
copy of $\Hh$ and $P^{(2)}$ on the second. Since $Q^{(1)}$ and $P^{(2)}$
live on \emph{different} tensor factors, they commute exactly,
$[Q^{(1)},P^{(2)}]=0$, in sharp contrast to $[Q,P]=\ii \id_{\mathcal{H}}$ within a single
copy. There is therefore no Heisenberg obstruction between them, and
they admit a generalized joint eigenfamily
$|q\rangle\otimes|{-p}\rangle$ for every $q,p\in\mathbb R$ --- which is
exactly $\Phi=|q\rangle\langle p|$ read back as an operator. The equation
for $\Psi$ is the conjugate statement,
$(Q^{(1)}+\ii P^{(2)})|\Psi\rangle\rangle=\lambda|\Psi\rangle\rangle$,
with eigenbasis $|q\rangle\otimes|p\rangle$, i.e.\ $\Psi=|q\rangle\langle -p|$.

This is the precise sense in which the present result differs from, yet
shares the same eigenvalue $\lambda=q+\ii p$ as, the 
coherent-state case: $a=(Q+\ii P)/\sqrt2$ is a genuinely
\emph{single}-mode operator, subject to the full uncertainty relation,
so an equation $a\Phi=\mu\Phi$ can only be satisfied by the best
available compromise, the coherent state $|\mu\rangle$; whereas
$\mathcal D\diamond_\sigma V=\lambda V$, once correctly derived, is a
\emph{two}-mode equation whose two halves never interact, and can
therefore be satisfied exactly, by the classical phase-space point
itself. Exchanging $\diamond_R\leftrightarrow\diamond_L$ exchanges which
tensor factor carries $Q$ and which carries $P$, matching the
ket$\leftrightarrow$bra exchange already noted above. In this language
the higher-order distributional solutions in the earlier remark are
generalized eigenvectors of the combined operator
$Q^{(1)}-\ii P^{(2)}$, not necessarily simultaneous eigenvectors of
$Q^{(1)}$ and $P^{(2)}$ separately --- a phenomenon
already familiar from the spectral theory of rigged Hilbert spaces, and
the natural framework in which to pursue Problem~\ref{prob:classify}.

\subsection*{A one-dimensional Fourier transform makes it explicit}

The two-mode picture above can be made completely concrete --- and the
proof of Proposition~\ref{prop:eigPhi} correspondingly elementary --- by
an ordinary one-dimensional Fourier transform in the \emph{second}
kernel variable alone. For $\Phi:\mathcal S(\mathbb R)\to\mathcal
S'(\mathbb R)$ with kernel $\Phi(x,y)$, set
$$
\hat\Phi(x,k)=\int_{\mathbb R}\Phi(x,y)\,\ex^{-\ii ky}\,\ud y.
$$
Since $\partial_y\mapsto \ii k$ under this transform, the kernel
equation $\partial_y\Phi=(x-\lambda)\Phi$ established in the proof of
Proposition~\ref{prop:eigPhi} becomes simply
$$
\boxed{(x-\ii k)\,\hat\Phi(x,k)=\lambda\,\hat\Phi(x,k),}
$$
an ordinary eigenvalue equation for a \emph{multiplication} operator on
$L^2(\mathbb R^2)$. Its solution is immediate: $\hat\Phi$ must be
supported at the unique point where $x-\ii k=\lambda$, i.e., writing
$\lambda=q+\ii p$, at $(x,k)=(q,-p)$ --- recovering
$\Phi=|q\rangle\langle p|$ without solving any differential equation.

Identifying $\mathbb R^2_{(x,k)}$ with $\mathbb C$ via $z=x+\ii k$, the
operator $Q\Phi+\ii\Phi P$ becomes, after this single transform,
multiplication by $\bar z$; the rank-one eigenoperators are point masses
$\delta^{(2)}(z-\bar\lambda)$, and $Q,P$ have literally become the two
components of the ordinary position operator on $L^2(\mathbb R^2)$. In
this precise sense the interlaced eigenvalue problem is a quantum
mechanics built directly on the classical phase plane $(q,p)$ itself, in
the spirit of the phase-space wave-function formalism of Torres-Vega and
Frederick~\cite{TorresVegaFrederick1990}, of the common-Liouville-space
approach of Ali and Prugove\v{c}ki~\cite{AliPrugovecki1977}, and of the
broader programme of coherent states, covariant phase-space observables
and joint position--momentum measurements developed by Ali, Beneduci and
their collaborators~\cite{AliAntoineGazeau2014,BeneduciSchroeck2013}.

This also sharpens the Jordan-chain remark above: transforming the
$k=1$ solution the same way gives
$$
\hat\Phi_{1} \propto c_0\,\delta^{(2)}(z-\bar\lambda)
+c_1\,\partial_z\delta^{(2)}(z-\bar\lambda),
\qquad
\partial_z=\tfrac12(\partial_x-\ii\partial_k),
$$
a derivative in the \emph{holomorphic} direction only, not the generic
two-parameter grid of $\partial_x,\partial_k$ derivatives a
point-supported distribution in $\mathbb R^2$ would allow in general ---
because the operator in question is multiplication by $\bar z$
specifically, not by $x$ or $k$ separately. This identifies
Problem~\ref{prob:classify} with the classical ``division problem'' for
the holomorphic tower $\partial_z^j\delta^{(2)}(z-z_0)$, $j\ge0$, a
promising route to a complete answer.

\begin{problem}
\label{prob:classify}
Classify completely the solution space of $Q\Phi+\ii\Phi P=\lambda\Phi$ in
$\mathcal S'(\mathbb R^2)$: is it exactly the Jordan-chain tower generated
by $\delta^{(k)}(x-q)\,\ex^{-\ii py}$, $k\ge0$, or are there further
solutions not of this form?
\end{problem}

\begin{problem}
Generalise the eigenvalue problem to
$\mathcal D\diamond_R\mathcal D\diamond_R V=\lambda^2 V$
using Proposition~\ref{prop:squares} and equation~\eqref{eq:eigR}, and
interpret the result in terms of the harmonic oscillator number
eigenstates.
\end{problem}


\section{Outlook}
\label{sec:outlook}

Five further directions seem worth recording briefly. First, the
non-associativity noted in Section~\ref{sec:structure} raises the natural
question of a universal associative envelope generated by the left and
right translations of $\diamond_R,\diamond_L$; for operator entries such
an envelope would presumably sit inside $B(\Hh)$, but making this precise
is left open. Second, the two-parameter interpolation
$X\diamond_\theta Y=\cos\theta\,X\diamond_RY+\sin\theta\,X\diamond_LY$
between the two orientations (Appendix~\ref{app:deformation}) suggests a
one-parameter deformation of the factorisation of
Section~\ref{sec:factorisation}, and, more speculatively, a connection
with the broader programme of non-commutative
geometry~\cite{Connes}.
Third, the two-mode interpretation of
Section~\ref{sec:eigen}, based on the vectorization
of operators, emerges naturally from the two-sided
Sylvester equations obtained in the interlaced
eigenvalue problem.

While the interlaced operator matrix acts naturally
on $\Hh\oplus\Hh$, its operator-valued eigenvalue
problem leads to a tensor-product representation
through the independent left and right actions
of $Q$ and $P$.

It would be interesting to determine whether this
tensor-product structure can be derived intrinsically
from the interlaced algebraic construction, independently
of the subsequent formulation of the eigenvalue problem.

Such an intrinsic understanding might also shed light
on the higher-power problem discussed in
Section~\ref{sec:eigen} and on the classification
of generalized eigenoperators raised in
Problem~\ref{prob:classify}.

Fourth, the mechanism behind this commutativity is completely
general: for \emph{any} $Q,P$ whatsoever, left multiplication
$\mathcal L_Q$ and right multiplication $\mathcal R_P$ on the space of
operators commute, since $(QX)P=Q(XP)$, regardless of $[Q,P]$. This is
the same structural move underlying Naimark's dilation
theorem~\cite{Naimark1943}: a POVM built from non-commuting effects is
always the compression of a genuine, commuting, projection-valued
measurement on a larger space. Its classical realisation for position
and momentum, the Arthurs--Kelly simultaneous-measurement
scheme~\cite{ArthursKelly1965}, and the rigorous treatment of Naimark
dilations for covariant phase-space and position--momentum
observables~\cite{BeneduciFrionGazeauPerri2022,BeneduciSchroeck2013},
reach the same doubled, commuting structure from an entirely
different, measurement-theoretic starting point. We record this
resonance without claiming a term-by-term identification.

Fifth, since the proof of Proposition~\ref{prop:squares} uses no
self-adjointness or commutation property of $Q,P$ beyond their
appearance in the defining formulas, the same interlaced factorisation
applies verbatim to non-Hermitian Hamiltonians built from $Q,P$ by
formal substitution --- for instance the exactly solvable
$\mathcal{PT}$-symmetric oscillator
$H_{\mathcal{PT}}=\tfrac12(P^2+Q^2)+\ii\alpha Q$, $\alpha\in\mathbb R$,
factorises identically with $Q$ replaced everywhere by the
complex-shifted $Q+\ii\alpha\id_{\mathcal H}$. A systematic treatment of
interlaced factorisations for basic $\mathcal{PT}$-symmetric
Hamiltonians is left for future work.

These five questions are left for future investigation.

\appendix

\section{Two-orientation deformations}
\label{app:deformation}

The presence of two orientations suggests the interpolating product
$$
X\diamond_\theta Y
=
\cos\theta\, X\diamond_RY
+
\sin\theta\, X\diamond_LY,
$$
with $\diamond_\theta|_{\theta=0}=\diamond_R$ and
$\diamond_\theta|_{\theta=\pi/2}=\diamond_L$. One may also interpolate
between ordinary matrix multiplication and the interlaced product,
$$
X\diamond_tY=(1-t)XY+t(X\diamond_RY),\quad t\in[0,1].
$$

\begin{remark}
$\diamond_\theta$ does not stay within the class of ``interlaced''
products for $\theta\notin\{0,\pi/2\}$: a generic linear combination of
$\diamond_R$ and $\diamond_L$ no longer has the property that each entry
contains exactly two terms, each a product of entries of $X$ and $Y$ in
one definite order. As a deformation for the purpose of studying the
associator, however, it is well defined.
\end{remark}

\begin{problem}
Study the associator
$\mathcal A_\theta(X,Y,Z)=(X\diamond_\theta Y)\diamond_\theta Z-
X\diamond_\theta(Y\diamond_\theta Z)$.
At $\theta=\pi/4$, is there a simplification due to the symmetry
$\diamond_R\leftrightarrow\diamond_L$?
\end{problem}


\section*{Statements and Declarations}
\addcontentsline{toc}{section}{Statements and Declarations}

\subsection*{Funding}
No funding was received to assist with the preparation of this manuscript.

\subsection*{Competing interests}
The author declares that he has no financial or non-financial interests
directly or indirectly related to the work submitted for publication.

\subsection*{Data availability}
Data sharing is not applicable to this article, as no datasets were
generated or analysed during the current study.

\subsection*{Author contribution}
J.-P.\ Gazeau is the sole author of this manuscript and takes full
responsibility for its content.


\subsection*{Acknowledgements}

The author is indebted to Alice Barbara Tumpach for her insightful
question about the (non?)-uniqueness of the oriented interlaced product,
and to Catherine Lacape for her warm hospitality in Montplaisir (Marvejols, FR), where
part of this work was carried out. The note is dedicated to the memory
of Franciszek Hugon Szafraniec, whose rigour it tries, however modestly,
to honour.



\section*{Bibliography}
\addcontentsline{toc}{section}{Bibliography}

\begin{thebibliography}{99}

\bibitem{GazeauSzafraniec2011}
J.-P. Gazeau and F.~H. Szafraniec,
Holomorphic Hermite polynomials and a non-commutative plane,
J.\ Phys.\ A: Math.\ Theor.\ \textbf{44}, 495201 (2011).
\url{https://doi.org/10.1088/1751-8121/44/49/495201}

\bibitem{GazeauSzafraniec2016}
J.-P. Gazeau and F.~H. Szafraniec,
Three paths toward the quantum angle operator,
Ann.\ Phys.\ \textbf{375}, 16--35 (2016).
\url{https://doi.org/10.1016/j.aop.2016.09.010}

\bibitem{Szafraniec1998}
F.~H. Szafraniec,
Analytic models for the quantum harmonic oscillator,
in: \emph{Operator Theory for Complex and Hypercomplex Analysis}
(Mexico City, 1994), Contemporary Mathematics vol.~212,
eds.\ E.\ Ram\'irez de Arellano, N.\ Salinas, M.~V.\ Shapiro and
N.~L.\ Vasilevski, American Mathematical Society, Providence, RI,
1998, pp.~269--276.

\bibitem{Gazeau2026}
J.-P. Gazeau,
A Matrix Mnemonic for Octonions and Cayley--Dickson Tower,
Adv.\ Appl.\ Clifford Algebras \textbf{36}, 34 (2026).

\bibitem{Baez}
J.~C. Baez,
The octonions,
Bull.\ Amer.\ Math.\ Soc.\ \textbf{39}, 145--205 (2002).

\bibitem{Schafer}
R.~D. Schafer,
\emph{An Introduction to Nonassociative Algebras},
Academic Press, New York, 1966.

\bibitem{Okubo}
S. Okubo,
\emph{Introduction to Octonion and Other Non-Associative Algebras in Physics},
Cambridge University Press, Cambridge, 1995.

\bibitem{SpringerVeldkamp}
T.~A. Springer and F.~D. Veldkamp,
\emph{Octonions, Jordan Algebras and Exceptional Groups},
Springer, Berlin, 2000.

\bibitem{KadisonRingrose1}
R.~V. Kadison and J.~R. Ringrose,
\emph{Fundamentals of the Theory of Operator Algebras, Vol.~I},
Academic Press, 1983.

\bibitem{JordanVNW}
P. Jordan, J. von Neumann and E. Wigner,
On an algebraic generalization of the quantum mechanical formalism,
Ann.\ of Math.\ \textbf{35}, 29--64 (1934).

\bibitem{McCrimmon}
K. McCrimmon,
\emph{A Taste of Jordan Algebras},
Springer, New York, 2004.

\bibitem{Pflugfelder}
H.~O. Pflugfelder,
\emph{Quasigroups and Loops: Introduction},
Heldermann Verlag, Berlin, 1990.

\bibitem{ReedSimon2}
M. Reed and B. Simon,
\emph{Methods of Modern Mathematical Physics, Vol.~II: Fourier Analysis,
Self-Adjointness},
Academic Press, 1975.

\bibitem{Dirac}
P.~A.~M. Dirac,
The quantum theory of the electron,
Proc.\ Roy.\ Soc.\ Lond.\ A \textbf{117}, 610--624 (1928).

\bibitem{GelfandVilenkin4}
I.~M. Gelfand and N.~Ya. Vilenkin,
\emph{Generalized Functions, Vol.~4: Applications of Harmonic Analysis},
Academic Press, New York, 1964.

\bibitem{Bohm1978}
A. Bohm,
\emph{The Rigged Hilbert Space and Quantum Mechanics},
Lecture Notes in Physics vol.~78, Springer, Berlin, 1978.

\bibitem{AntoineTrapani2009}
J.-P. Antoine and C. Trapani,
\emph{Partial Inner Product Spaces:
Theory and Applications},
Lecture Notes in Mathematics, Vol.~1986,
Springer, Berlin, 2009.
\url{https://doi.org/10.1007/978-3-642-05136-4}

\bibitem{AntoineInoueTrapani2002}
J.-P. Antoine, A. Inoue and C. Trapani,
\emph{Partial *-Algebras and Their Operator Realizations},
Mathematics and Its Applications, Vol.~553,
Kluwer Academic Publishers, Dordrecht, 2002.
\url{https://doi.org/10.1007/978-94-017-0065-8}

\bibitem{Treves1967}
F. Tr\`eves,
\emph{Topological Vector Spaces, Distributions
and Kernels},
Academic Press, New York, 1967.

\bibitem{TorresVegaFrederick1990}
G. Torres-Vega and J.~H. Frederick,
Quantum mechanics in phase space: New approaches to the correspondence
principle,
J.\ Chem.\ Phys.\ \textbf{93}, 8862--8874 (1990).
\url{https://doi.org/10.1063/1.459225}

\bibitem{AliPrugovecki1977}
S.~T. Ali and E. Prugove\v{c}ki,
Classical and quantum statistical mechanics in a common Liouville space,
Physica A \textbf{89}, 501--521 (1977).

\bibitem{AliAntoineGazeau2014}
S.~T. Ali, J.-P. Antoine and J.-P. Gazeau,
\emph{Coherent States, Wavelets, and Their Generalizations},
2nd edn., Theoretical and Mathematical Physics, Springer, New York, 2014.

\bibitem{BeneduciSchroeck2013}
R. Beneduci and F.~E. Schroeck Jr.,
A note on the relationship between localization and the norm-1
property,
J.\ Phys.\ A: Math.\ Theor.\ \textbf{46}, 305303 (2013).

\bibitem{Connes}
A. Connes,
\emph{Noncommutative Geometry},
Academic Press, 1994.

\bibitem{Naimark1943}
M.~A. Naimark,
On a representation of additive operator set functions,
C.\ R.\ (Doklady) Acad.\ Sci.\ URSS (N.S.) \textbf{41}, 359--361 (1943).

\bibitem{ArthursKelly1965}
E. Arthurs and J.~L. Kelly Jr.,
On the simultaneous measurement of a pair of conjugate observables,
Bell Syst.\ Tech.\ J.\ \textbf{44}, 725--729 (1965).

\bibitem{BeneduciFrionGazeauPerri2022}
R. Beneduci, E. Frion, J.-P. Gazeau and A. Perri,
Quantum formalism on the plane: POVM-Toeplitz quantization, Naimark
theorem and linear polarization of the light,
Ann.\ Phys.\ \textbf{447}, 169134 (2022).

\end{thebibliography}
\end{document}